\documentclass[
	submission,
copyright,creativecommons,UKenglish]{eptcs}
\providecommand{\event}{QPL 2020}

\newcommand{\bw}{1}

\usepackage{algorithm}
\usepackage{algorithmicx}

\usepackage{amsmath,amsthm,amssymb}
\usepackage{mathtools}
\usepackage{xspace,enumerate,color,epsfig}
\usepackage{microtype}
\usepackage{graphicx}
\graphicspath{{.}{./figures/}}
\usepackage{marvosym}
\usepackage{bm}
\usepackage{tabu}
\usepackage{tikzit}

\if\bw1

\tikzstyle{Black dot}=[dot, fill={gray!40!white}, draw=black, shape=circle, tikzit fill={rgb,255: red,128; green,128; blue,128}, tikzit draw=black, tikzit shape=circle]
\tikzstyle{White dot}=[dot, fill=white, draw=black, shape=circle, tikzit fill=white, tikzit draw=black]
\tikzstyle{H}=[fill=white, draw, inner sep=0.6mm, minimum height=1.5mm, minimum width=1.5mm, tikzit shape=rectangle]
\tikzstyle{H box}=[H, tikzit shape=rectangle]
\tikzstyle{hadamard}=[H, tikzit shape=rectangle]
\tikzstyle{small hadamard}=[fill=white, draw, inner sep=0.6mm, minimum height=1.5mm, minimum width=1.5mm]
\tikzstyle{Function box}=[fill=white, draw=black, shape=rectangle, tikzit draw=black, tikzit fill=white, minimum width=1cm, minimum height=1cm]
\tikzstyle{dot}=[inner sep=0.3mm, minimum width=2mm, minimum height=2mm, draw, shape=circle, font={\footnotesize}]
\tikzstyle{white dot}=[dot, fill=white]
\tikzstyle{gray dot}=[dot, fill={rgb,255: red,128; green,128; blue,128}, text depth=-0.2mm]
\tikzstyle{cnot ctrl}=[fill=black, draw=black, shape=circle, inner sep=0pt, minimum width=1.2mm, tikzit category=circuit]
\tikzstyle{cnot targ}=[fill=white, draw=white, shape=circle, tikzit category=circuit, label={center:$\oplus$}, inner sep=0pt, minimum width=2.1mm, tikzit fill={rgb,255: red,102; green,204; blue,255}, tikzit draw=black]
\tikzstyle{gray phase dot}=[gray dot, tikzit draw={rgb,255: red,128; green,128; blue,128}]
\tikzstyle{white phase dot}=[minimum size=5mm, font={\footnotesize\boldmath}, shape=rectangle, rounded corners=2mm, inner sep=0.2mm, outer sep=-2mm, scale=0.8, tikzit shape=circle, draw=black, fill=white, tikzit draw=blue]
\tikzstyle{ehadamard}=[fill=white, draw, double, double distance=1pt, inner sep=0.9mm, minimum height=1.5mm, minimum width=1.5mm, tikzit shape=rectangle]

\tikzstyle{point}=[fill=white, draw=black, shape=regular polygon, regular polygon sides=3, scale=0.5, inner sep=0pt,minimum size=8mm]
\tikzstyle{left point}=[point, regular polygon rotate=-30]
\tikzstyle{right point}=[point, regular polygon rotate=30]
\tikzstyle{up point}=[point]
\tikzstyle{down point}=[point, regular polygon rotate=180]

\tikzstyle{braceedge}=[-,decorate, decoration={brace,amplitude=1mm,raise=-1mm}]
\tikzstyle{special edge}=[-, draw=blue, tikzit draw=blue]
\tikzstyle{diredge}=[->]

\input{zh.tikzdefs}
\fi

\if\bw0

\input{zh.tikzdefs}
\fi

\usepackage{stmaryrd}
\usepackage{docmute}
\usepackage{keycommand}

\usepackage{hyperref}
\usepackage[all]{hypcap} 

\theoremstyle{definition}
\newtheorem{theorem}{Theorem}[section]

\newtheorem{lemma}[theorem]{Lemma}
\newtheorem{proposition}[theorem]{Proposition}

\newtheorem{definition}[theorem]{Definition}

\newtheorem{example}[theorem]{Example}

\newtheorem{example*}[theorem]{Example*}
\newtheorem{examples*}[theorem]{Examples*}
\newtheorem{remark}[theorem]{Remark}
\newtheorem{remark*}[theorem]{Remark*}

\usepackage[color]{changebar}

\usepackage{color}
\def\bR{\begin{color}{red}} 
\def\bB{\begin{color}{blue}}
\def\bM{\begin{color}{magenta}}
\def\bC{\begin{color}{cyan}}
\def\bW{\begin{color}{white}}
\def\bBl{\begin{color}{black}} 
\def\bG{\begin{color}{green}}
\def\bY{\begin{color}{yellow}}
\def\e{\end{color}\xspace}
\newcommand{\bit}{\begin{itemize}}
\newcommand{\eit}{\end{itemize}\par\noindent}
\newcommand{\ben}{\begin{enumerate}}
\newcommand{\een}{\end{enumerate}\par\noindent}
\newcommand{\beq}{\begin{equation}}
\newcommand{\eeq}{\end{equation}\par\noindent}
\newcommand{\beqa}{\begin{eqnarray*}}
\newcommand{\eeqa}{\end{eqnarray*}\par\noindent}
\newcommand{\beqn}{\begin{eqnarray}}
\newcommand{\eeqn}{\end{eqnarray}\par\noindent}

\usepackage{proof}
\usepackage{dsfont}

\newcommand{\bra}[1]{\ensuremath{\left\langle #1 \right|}}
\newcommand{\ket}[1]{\ensuremath{\left|  #1 \right\rangle}}

\newcommand{\ketbra}[2]{\ensuremath{\ket{#1}\!\bra{#2}}}
\newcommand{\C}{\mathbb{C}}
\newcommand{\R}{\mathbb{R}}
\newcommand{\N}{\mathbb{N}}

\newcommand{\B}{\mathbb{B}}

\renewcommand{\vec}[1]{\mathbf{#1}}

\newcommand{\w}{\vec{w}}
\newcommand{\x}{\vec{x}}

\newcommand{\ie}{i.e.\xspace}

\renewcommand{\int}[1]{\left[#1\right]}

\newcommand{\lift}[1]{\overline{#1}}

\newcommand{\reduces}{\longrightarrow}

\newcommand{\pz}[1]{[#1]_{\textrm{ps}\to\textrm{zh}}}
\newcommand{\zp}[1]{[#1]_{\textrm{zh}\to\textrm{ps}}}

\title{Hypergraph Simplification: \\Linking the Path-sum Approach to the ZH-calculus}
\def\titlerunning{Hypergraph Simplification}

\author{Louis Lemonnier
\institute{ENS Paris-Saclay, Universit\'e Paris-Saclay}
\email{louis.lemonnier@ens-paris-saclay.fr} \and
    John van de Wetering
\institute{Radboud Universiteit Nijmegen}
\email{wetering@cs.ru.nl} \and
	Aleks Kissinger
\institute{Oxford University}
\email{aleks.kissinger@cs.ox.ac.uk}
}
\def\authorrunning{L. Lemonnier, J. van de Wetering \& A. Kissinger}

\begin{document}

\maketitle

\begin{abstract}
The ZH-calculus is a complete graphical calculus for linear maps between qubits that admits a straightforward encoding of hypergraph states and circuits arising from the Toffoli+Hadamard gate set. 
In this paper, we establish a correspondence between the ZH-calculus and the path-sum formalism, a technique recently introduced by Amy to verify quantum circuits. In particular, we find a bijection between certain canonical forms of ZH-diagrams and path-sum expressions. We then introduce and prove several new simplification rules for the ZH-calculus, which are in direct correspondence to the simplification rules of the path-sum formalism. The relatively opaque path-sum rules are shown to arise naturally from two powerful families of rewrite rules in the ZH-calculus. The first is the extension of the familiar graph-theoretic simplifications based on local complementation and pivoting to their hypergraph-theoretic analogues: hyper-local complementation and hyper-pivoting. The second is the graphical Fourier transform introduced by Kuijpers et al., which enables effective simplification of ZH-diagrams encoding multi-linear phase polynomials with arbitrary real coefficients.
\end{abstract}

\section{Introduction}

The very nature of quantum computation makes it hard to verify classically that a given quantum circuit implements the desired computation without incurring exponential space or time costs.
It is however still possible to develop smart heuristics that can verify that quantum circuits indeed implement the right unitary.
One such heuristic is the \emph{path-sum} approach~\cite{Amy18}.
It represents each quantum gate in the circuit by the action it has on the computational basis states, given by a Boolean function determining the output basis state and a semi-Boolean function giving relative phases of outputs, each of which can depend on inputs as well as auxiliary Boolean variables, or `paths', that are summed over. 
Amy developed a set of simplification rules for these path-sums that were powerful enough to completely simplify a set of benchmark quantum circuits that implemented classical reversible functions to their classical specification.
Each of these rewrite rules eliminates a variable from the path-sum, but beyond that, their interpretation is quite opaque.

In this paper we will see that path-sum expressions and the rewrite rules from~\cite{Amy18} can be represented in a natural way using the \emph{ZH-calculus}.
The ZH-calculus is a graphical language recently introduced by Backens and Kissinger~\cite{BK18} that can straightforwardly represent computations involving Hadamard and Toffoli gates, and generalisations thereof. It comes with a set of graphical rewrite rules that are \emph{complete}, meaning that any two diagrams representing equal linear maps can be graphically transformed into one another.

There are two key ingredients in our translation of the path-sum rewrite rules into the ZH-calculus. The first is based on the realisation that ZH-diagrams can easily represent \emph{hypergraph states}~\cite{rossi2013quantum,qu2013encoding}.
\emph{Graph states} are a type of stabiliser state widely used in a variety of quantum protocols as well as the one-way model of measurement-based quantum computation~\cite{MBQC2}. Interestingly, the graph operations of \emph{local complementation} and \emph{pivoting} can be performed on the underlying graph of a graph state by applying local Cliffords~\cite{NestMBQC}.
We will show in this paper that the graph-theoretic simplifications of ZX-diagrams from~\cite{DKPW19} based on local complementation and pivoting extend naturally to hypergraph-theoretic simplifications of ZH-diagrams, which we call \textit{hyper-local complementation} and \textit{hyper-pivoting}.

The second ingredient is producing a ZH-calculus analogue to the operation of \textit{lifting} of a Boolean polynomial $Q$ to a related polynomial $\overline{Q}$ to enable substitution into a semi-boolean function. This operation plays an important role in the path-sum reductions, and it turns out to correspond to the \textit{graphical Fourier transform} introduced by Kuijpers and two of the authors in~\cite{KWK19}.

By introducing a new connection between path-sums and hypergraph states via the ZH-calculus we provide not only a new perspective on verification of quantum circuits, but also new techniques that can be applied to measurement-based quantum computing models based on hypergraph states~\cite{gachechiladze2019changing,takeuchi2019quantum,gachechiladze2019quantum}. Computations could be analysed and verified using either the ZH-calculus or path-sums, similar to how the ZX-calculus can be used to analyse the one-way model of MBQC~\cite{backens2020extraction,duncan2010rewriting}.




\noindent \textbf{Related work.} This work is an extension of a technical report by one of the authors in 2019~\cite{lemonnierreport}. Since then, parallel work of Vilmart has drawn a similar correspondence between the path-sum approach and the ZH-calculus~\cite{vilmart2020structure}. That work differs from ours in two ways. First, its correspondence is semantical: i.e. it defines categories of ZH-diagrams and sum-over-paths expressions, modulo the appropriate laws and gives functors in either direction. Ours is syntactic: we establish a direct bijection between certain universal families of ZH-diagrams and path-sum expressions. 
Second, at the level of rewriting, Ref.~\cite{vilmart2020structure} focuses on the Clifford fragments of each of the two languages and proves completeness, whereas we aim to capture the full power the two languages in reasoning beyond Clifford computations, but like~\cite{Amy18}, our goal is to develop useful heuristics for diagram simplification rather than a complete procedure for deciding equality (which is a QMA-hard problem~\cite{bookatz2012qmacomplete}).

Much as the ZX-calculus is closely connected to the theory of graph states and the one-way model, the ZH-calculus is closely connected to the emerging theory of hypergraph states. In that context, a notion of local complementation for hypergraph states has been introduced by Gachechildaze et al.~\cite{gachechiladze2017graphical} which is closely related to the simplification we introduce in Section~\ref{sec:hlc}. In addition, a restricted version of the graphical Fourier transform was described for \emph{weighted} hypergraph states in \cite{tsimakuridze2017graph}.


\section{ZH-calculus}
\label{s:ZH-dfn}

In this section we will recall the ZH-calculus, together with (annotated) !-box notation and the Fourier transform rule of \cite{KWK19}.

The ZH-calculus is a diagrammatic language introduced by Backens and Kissinger~\cite{BK18} that represents linear maps as \emph{ZH-diagrams}. These are string diagrams based on two generators: Z-spiders, depicted as white dots; and H-boxes, depicted as white boxes with a complex parameter $a$:
$$
\input{Z-spider.tikz} := \ket{0\ldots0}\bra{0\ldots0} + \ket{1\ldots1}\bra{1\ldots1} \qquad\qquad
\input{H-spider.tikz} := \sum a^{i_1\ldots i_m j_1\ldots j_n} \ket{j_1\ldots j_n}\bra{i_1\ldots i_m}
$$
Here in the right-hand equation the sum runs over all $i_1,\ldots, i_m, j_1,\ldots, j_n\in\{0,1\}$. Hence, an H-box represents a matrix with $a$ as its $\ket{1\ldots1}\bra{1\ldots1}$ entry, and ones everywhere else. By convention, we omit the parameter $a$ when $a=-1$, and hence an unlabeled H-box with 1 input and 1 output is the conventional Hadamard gate (up to normalisation).
More complex diagrams are constructed by composing these generators either by stacking them or by joining the outputs of the first with the inputs of the second, which correspond respectively to the tensor product and regular composition of linear maps.

Our calculations will greatly benefit from the usage of !-box -- pronounced as ``bang box''
-- notation~\cite{KMS12}. A !-box, drawn as a blue square around a piece of a ZH-diagram, represents a part of the diagram that may be replicated an 
arbitrary number of times, and hence allows one to express a whole family of diagrams at once:

$$
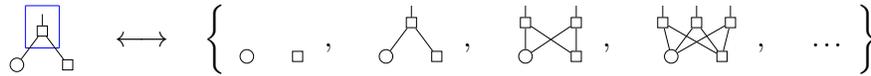
\input{bang-box-example.tikz} \quad \longleftrightarrow \quad
 \left\{
 \ \ \begin{tikzpicture}
	\begin{pgfonlayer}{nodelayer}
		\node [style=white dot] (0) at (-0.75, -0.5) {};
		\node [style=H] (7) at (0.75, -0.5) {};
	\end{pgfonlayer}
\end{tikzpicture}
\ \ ,\quad
 \ \ \begin{tikzpicture}
	\begin{pgfonlayer}{nodelayer}
		\node [style=white dot] (0) at (-0.75, -0.5) {};
		\node [style=hadamard] (1) at (0, 0.5) {};
		\node [style=none] (2) at (0, 1) {};
		\node [style=H] (7) at (0.75, -0.5) {};
	\end{pgfonlayer}
	\begin{pgfonlayer}{edgelayer}
		\draw (2.center) to (1);
		\draw (1) to (0);
		\draw (7) to (1);
	\end{pgfonlayer}
\end{tikzpicture}
\ \ ,\quad
 \ \ \input{bang-box-example2.tikz}\ \ ,\quad
 \ \ \input{bang-box-example3.tikz}\ \ ,\quad
 \ \ \ldots\ \  \right\}
$$

When used in equations, corresponding !-boxes on either side of the equation should be understood to be replicated 
an equal number of times.

As in \cite{BK18} we will also use \emph{annotated} !-boxes that are labelled by a set or a natural number to denote the number of copies of the diagram. For example, letting $\B = \lbrace 0,1 \rbrace$ denote the set of Booleans we have:
$$\input{figures/annotatedIntro.tikz} \qquad\quad\text{and}\qquad\quad \input{figures/overlapIntro.tikz}$$
Note that in the right-hand diagram we had overlapping !-boxes resulting in a fully-connected bipartite graph of connectivity.
Also following \cite{BK18} we use some derived generators: the X-spider and the NOT gate.

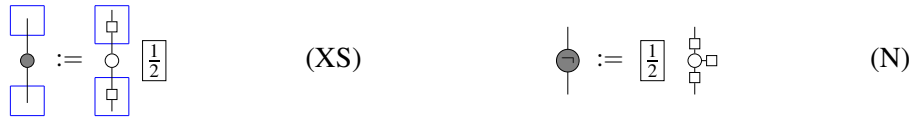
\begin{minipage}{.45\linewidth}
\beq\tag{XS}\label{eq:grey-spider}
 \input{X-spider-dfn-bb.tikz}
\eeq
\end{minipage}
\begin{minipage}{.45\linewidth}
\beq\tag{N}\label{eq:X-dfn}
 \input{negate-dfn.tikz}
\eeq
\end{minipage}

The power of the ZH-calculus comes from the set of graphical rewrite rules associated to it. First of all, ZH-diagrams are considered equal when they can be topologically deformed into one another, as long as the order of in- and outputs is preserved~\cite{CD09,PQP}. Second, there are a set of rewrite rules that can be applied to parts of a diagram. We present these standard rules in Figure~\ref{fig:ZH-rules}. 

\begin{figure}[!htb]
 \centering
 \scalebox{0.9}{
 \begin{tabular}{|cccccccc|}
 \hline 
 &&&&&&&\\
 (ZS1) & \input{Z-spider-rule-bb.tikz} & \qquad \quad & (ZS2) & \input{Z-special.tikz} & \qquad & (BA1) & \input{ZX-bialgebra-bb.tikz} \\ &&&&&&&\\
 (HS1) & \input{H-spider-rule-bb.tikz} & & (HS2) & \input{H-identity.tikz} & & (BA2) & \input{ZH-bialgebra-bb.tikz} \\ &&&&&&&\\
 (M) & \input{multiply-rule-bb.tikz} & & (U) & \input{unit-bangboxed.tikz} & & (I) & \input{intro-rule-bangboxed.tikz} \\ &&&&&&&\\
 (A) & \input{average-rule.tikz} & & (O) & \input{ortho-rule.tikz} & & &\\ &&&&&&&\\
 \hline
 \end{tabular}}
 \caption{The rules of the ZH-calculus.
 Throughout, $a,b$ are arbitrary complex numbers. These are the !-boxed versions of the rules as presented in \cite{BK18}.
 \label{fig:ZH-rules}}
\end{figure}

These rules are \emph{complete} meaning that if two ZH-diagrams represent the same linear map, then those diagrams can be transformed into one another using some application of these rules~\cite{BK18}.





Some of our results will require the \emph{Fourier transform} of a ZH-diagram as constructed in \cite{KWK19}. Before we introduce this Fourier transform, we recall some of the notation of \cite{KWK19}. First, there are the \emph{exponentiated} H-boxes and the associated phase spiders, that allow us to make the connection between ZH-diagrams and path-sums more direct:
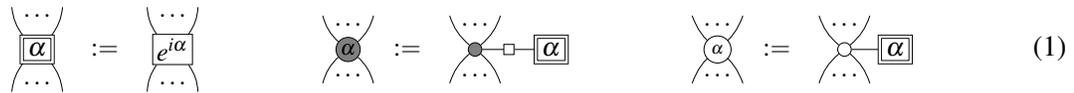
\begin{equation}\label{eq:phase-spider-dfn}
    \input{eH-spider.tikz}  \qquad\qquad \input{X-spider-dfn-phase.tikz} \qquad\qquad \input{Z-spider-dfn-phase.tikz}
\end{equation}
Second, there are the \emph{disconnect boxes} from \cite{KWK19} that combine nicely with annotated !-boxes:
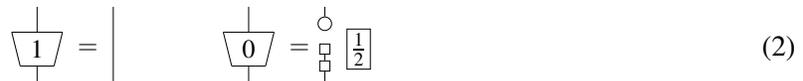
\begin{equation}\label{eq:disconnect-box}
\input{figures/bNotation.tikz}
\end{equation}
Let $[n] = \lbrace 1, \dots, n\rbrace$ denote the $n$ element set and let $\lvert \vec{b} \rvert$ denote the \emph{weight} of a bitstring $\vec b\in\mathbb{B}^n$, i.e.~the number of $1$'s in $\vec b$. Denote by $\mathbb{B}^n_*$ the set of all non-zero bitstrings, i.e.~$\mathbb{B}^n\setminus (0,\dots,0)$.
\begin{proposition}\label{prop:fourier-transform}
\cite{KWK19} The following Fourier-transform rule holds.
$$\input{fourierBang-exp.tikz}$$
\end{proposition}
For example, in the $n=2$ and $n=3$ cases we have:
$$\input{fourier-example-2.tikz} \qquad\qquad \input{fourier-example-3.tikz}$$

\begin{remark}
	The reason we call this rule a Fourier transform is because it is closely related to the Fourier transform of semi-Boolean functions. See \cite{KWK19} for the details. 
	We presented the rule using exponentiated H-boxes. A more general version using regular H-boxes (as long as the label is non-zero) also holds. 
\end{remark}

Before we continue onto our introduction of path-sums, it will be useful to introduce a new class of ZH-diagrams that generalises the definition of graph-like ZX-diagrams from \cite{DKPW19}.

\begin{definition}
	We say a ZH-diagram is \emph{hypergraph-like} when 
	\begin{itemize}
		\item all spiders are Z-spiders,
		\item every in- and output wire is connected only to a spider (so no connections directly to H-boxes),
		\item the only wires are between H-boxes and spiders,
		\item there is at most one wire between any given H-box and spider (so no parallel edges),
		\item and there are no H-boxes connected to exactly the same set of spiders.
	\end{itemize}
\end{definition}

We call these diagrams hypergraph-like, because most of their structure is captured by the underlying hypergraph that has as vertices the spiders and as hyperedges the H-boxes.
\begin{definition}
	A \emph{hypergraph} $G=(V,E)$ consists of a set of \emph{vertices} $V$ and a set of \emph{hyperedges} $E$. Each hyperedge $e\in E$ is a non-empty set of vertices $e\subseteq V$. We call a hyperedge \emph{simple} when it contains exactly two vertices. A \emph{simple graph} is a hypergraph where every hyperedge is simple.
\end{definition}
The underlying hypergraph of a hypergraph-like ZH-diagram is a simple graph iff all H-boxes have arity 2. Such diagrams are \emph{graph-like} as defined in \cite{DKPW19}.
Every ZH-diagram can be reduced to a hypergraph-like ZH-diagram representing the same linear map. Before we prove this, we need a lemma.

\begin{lemma}\label{lem:sv}
Multiple parallel edges between an H-box and a Z-spider can be reduced to a single edge.
        \beq\label{eq:sv}
        \input{reduce-parallel-wires.tikz}
        \eeq
\end{lemma}
\begin{proof}
	Follows easily from Lemma~5.1 in~\cite{backens2021completeness}.  
\end{proof}

\begin{lemma}\label{lem:ZH-to-hypergraph}
	Every ZH-diagram can be efficiently transformed into a hypergraph-like ZH-diagram.
\end{lemma}
\begin{proof}
	First transform all grey spiders and NOT gates into white spiders using the definitions (XS) and (N). Remove all double Hadamard gates this introduces with (HS2) --- the equation label refers to Figure~\ref{fig:ZH-rules}. Fuse all the spiders by applying (ZS1) repeatedly. Disconnect in- and outputs from H-boxes by introducing identities with (ZS2), and similarly introduce identities between H-boxes that are connected. Remove parallel edges between H-boxes and spiders with Lemma~\ref{lem:sv}. Finally, fuse H-boxes that have the same set of neighbours using (M).
\end{proof}

\begin{remark}
	The \emph{weighted} hypergraphs of \cite{tsimakuridze2017graph} associate phases to each hyperedge. The resulting weighted hypergraph states precisely correspond to hypergraph-like ZH-diagrams where every spider has exactly one output wire and there are no inputs. The ZH normal form diagrams of \cite{BK18} similarly have each Z-spider connected to a unique output, but there H-boxes are also allowed to be connected to white spiders via a NOT gate.
\end{remark}

\section{Path-sums and pure path-sums}
\label{s:pps}

Path-sums give a compact way of representing the action of a linear map $A$ on computational basis states in terms of two polynomial functions $f$ and $\phi$:
\begin{equation}\label{eq:amy-path-sum}
A :: \ket{\bm x} \mapsto \lambda \cdot \sum_{\bm y} e^{2\pi i \cdot \phi(\bm x, \bm y)} \ket{f(\bm x, \bm y)}
\end{equation}
In this description $\lambda \in \mathbb C$ is a (typically irrelevant) global scalar factor; $\bm x \in \mathbb B^n$ where $n$ is the number of input qubits of $A$ is a bit string which we will typically, by minor abuse of notation, treat as lists of variables; the sum is over all bitstrings $\bm y \in \mathbb B^k$ for some $k$; 
$f = (f_1, \ldots, f_m)$ where each $f_i \in \mathbb B[\bm x, \bm y]$ is a Boolean (i.e. $\mathbb F_2$-valued) polynomial describing the $i$-th output basis state in terms of $\bm x$ and $\bm y$ and $m$ is the number of qubit outputs of $A$;
and $\phi \in \mathbb R[\bm x, \bm y]$ is a polynomial valued in the real numbers (or some sub-ring thereof) describing the phase of each summand, which is often referred to as the \textit{phase polynomial}.

\begin{example}\label{ex:path-sum-CNOT}
	The path-sum representation of the CNOT gate is $\ket{x_1 x_2} \mapsto \ket{x_1 (x_1\oplus x_2)}$. Hence, we have $\phi = 1$, no path variables $\vec y$, and $f=(f_1,f_2)$ where $f_1(x_1,x_2)=x_1$ and $f_2(x_1,x_2) = x_1\oplus x_2$.
\end{example}

While there are some practical benefits for keeping the data $f$ and $\phi$ separate, we will consider a slight variation on path-sum expressions that we call \textit{pure path-sum expressions}, which keeps all of the relevant data about $U$ in $\phi$ and treats inputs and outputs on the same footing. This will enable us to make an exact correspondence with ZH-diagrams in the sequel.

For a set $S$, let $S^*$ be the set of finite lists of elements of $S$.

\begin{definition}
A \textit{pure path-sum expression} consists of:
\begin{itemize}
\item A set of \textit{path variables} $\bm x = (x_1, \ldots, x_k)$,
\item an \textit{input signature} $\vec i \in \{1,\ldots k\}^*$,
\item an \textit{output signature} $\vec o \in \{1, \ldots k\}^*$,
\item and a \textit{phase polynomial} $\phi(\bm x) \in \R[\bm x]$.
\end{itemize}
\end{definition}

The \emph{associated linear operator} of a pure path-sum expression is
$
A := \sum_{\bm x} e^{2\pi i\phi(\bm x)} \ket{\bm x_{\vec o}}\bra{\bm x_{\vec i}}
$
where the sum is over all bitstrings $\bm x \in \mathbb{B}^k$
and $\bm x_{\vec o} = x_{o_1}x_{o_2}\cdots x_{o_{\lvert o \rvert}}$ and $\bm x_{\vec i} = x_{i_1}x_{i_2}\cdots x_{i_{\lvert i \rvert}}$. I.e.~$\bm x_{\vec o}$ contains only the components of $\bm x$ that are listed in the output signature $\vec o$ and can have repetitions and permutations of the elements.

We don't lose any expressiveness by using pure path-sum expressions. In fact, we can translate a path-sum expression in the form of~\eqref{eq:amy-path-sum} to a pure path-sum expression, at the cost of introducing some dummy variables:
\begin{equation}\label{eq:dummy-vars}
\lambda \sum_{\bm x, \bm y} e^{2\pi i \cdot \phi(\bm x)} \ket{f(\bm x, \bm y)}\bra{\bm x}
 =
\frac{\lambda}{2^m} \sum_{\bm v, \bm w, \bm x, \bm y} e^{2\pi i \cdot \big[\phi(\bm x) + \frac{1}{2}\sum_j v_j (w_j + f_j(\bm x, \bm y))\big]} \ketbra{\bm w}{\bm x}.
\end{equation}
Here we used the identity $\frac 1 2 \sum_{v_j} e^{i \pi v_j(w_j + f_j(\bm x, \bm y))} = \delta_{w_j, f_j(\bm x, \bm y)}$ to obtain the RHS above.

\begin{remark}\label{remark:monomials}
	It will be convenient in the sequel for $\phi$ to have a certain canonical form. 
	Namely that $\phi(\bm x) = \sum_i \lambda_i \phi_i(\bm x)$ for 
	$\lambda_i \in \R$ where each $\phi_i$ is a \emph{Boolean monomial}, 
	\ie $\phi_i(\bm x) = \prod_j x_j\cdot y_{i,j}$ for some $\vec{y_i} \in \mathbb{B}^k$.
	The procedure above introduces terms in $\phi$ that are not of this form, but are $e^{2\pi i \alpha f(\bm x)}$ where $f:\mathbb{B}^k \rightarrow \mathbb{B}$ is some Boolean function. Such functions $f$ can however always be written as an XOR of monomials. Using the identity $x\oplus y = x + y - 2x\cdot y$ repeatedly we can then write any Boolean function $f$ as $f=\sum_j \lambda_j f_j$ where $\lambda_j\in \R$ and $f_j$ are monomials. Hence $e^{2\pi i \alpha f(\bm x)} = e^{2\pi i \sum_j \alpha\lambda_j f_j(\vec x)}$
	is of the desired form.
\end{remark}

\begin{example}
	Applying the transformation of Eq.~\eqref{eq:dummy-vars} to the path-sum representation of Example~\ref{ex:path-sum-CNOT} yields
	$$\text{CNOT} \ = \ \ \frac14 \ \sum_{\mathclap{\substack{v_1,v_2,\\w_1,w_2,\\x_1,x_2}}} e^{2\pi i \cdot \frac 1 2 \big[v_1(w_1+x_1) + v_2(w_2+(x_1\oplus x_2))\big]} \ketbra{w_1,w_2}{x_1,x_2}.$$
	Further applying $x_1\oplus x_2 = x_1 + x_2 - 2x_1\cdot x_2$, and eliminating monomials with coefficient $1$, we can simplify this to
	\begin{equation}\label{eq:cnot-path-sum}
	\text{CNOT} \ = \ \ \frac12 \ \sum_{\mathclap{\substack{v_1,v_2,\\w_1,w_2,\\x_1,x_2}}} e^{2\pi i \cdot \big[
	\frac 1 2 v_1w_1 +
	\frac 1 2 v_1x_1 +
	\frac 1 2 v_2w_2 +
	\frac 1 2 v_2x_1 +
	\frac 1 2 v_2x_2
	\big]} \ketbra{w_1,w_2}{x_1,x_2},
	\end{equation}
	which is indeed a pure path-sum.
	It can actually be further simplified to the path-sum of Eq.~\ref{eq:CNOT-path-sum-ZH} later on by using the identity $\sum_{v_1,w_1}(-1)^{v_1w_1+v_1x_1}\ket{w_1} = 2\ket{x_1}$.
\end{example}

A pure path-sum allows easy repetition of variables in the input as well as the output, so that we can succinctly write linear maps such as the following one representing a Z-spider with 2 inputs and 1 output:
\[
A
\ :=\ 
\sum_{x} e^{2\pi i \cdot 0} \ket{x}\bra{xx}
\]
While it is possible to write these maps using a standard path-sum, this requires additional dummy variables, e.g.
$
M
\ ::\ \ket{x_0 x_1} \mapsto
\frac 1 2 \sum_{v} e^{2\pi i \cdot \big[ \frac 1 2 (v (x_0 + x_1)) \big]} \ket{x_0}
$.


In \cite{Amy18}, several reduction rules were presented for path-sum expressions. Each of these 4 rules removes at least one path-variable from the expression. We present these rules, translated into pure path-sum expressions, in Figure~\ref{fig:rewrite}.
\begin{figure}[!tb]
	\begin{align*}
	\lambda\sum_{y_0,\bm x}  e^{2\pi i R(\bm x)} \ket{\bm x_{\vec o}}\bra{\bm x_{\vec i}} 
		&\reduces 2\lambda \sum_{\bm x} e^{2\pi i R(\bm x)} \ket{\bm x_{\vec o}}\bra{\bm x_{\vec i}}  \qquad\qquad\qquad\qquad\qquad\textsf{[Elim]}\\
	\lambda\sum_{y_0,\bm x}  e^{2\pi i \left(\frac{1}{4}y_0 + \frac{1}{2}y_0 Q(\bm x) + R(\bm x)\right)} \ket{\bm x_{\vec o}}\bra{\bm x_{\vec i}} 
		&\reduces \sqrt{2}\lambda\sum_{\bm x}e^{2\pi i\left(\frac{1}{8} -\frac{1}{4}\lift{Q}(\bm x) + R(\bm x)\right)}\ket{\bm x_{\vec o}}\bra{\bm x_{\vec i}} \qquad\qquad\quad\textsf{[$\omega$]}\\
	\lambda\sum_{y_0,y_1,\bm x}  e^{2\pi i \left(\frac{1}{2}y_0(y_1 + Q(\bm x)) + R(y_1,\bm x)\right)}\ket{\bm x_{\vec o}}\bra{\bm x_{\vec i}} 
		&\reduces 2\lambda\sum_{\bm x}   e^{2\pi i\left(R[y_1\gets \lift{Q}]\right)(\bm x)}\ket{\bm x_{\vec o}}\bra{\bm x_{\vec i}} \qquad\qquad\qquad \quad \textsf{[HH]}
	\end{align*}
	\vspace{-0.6cm}
	\begin{align*}
		{}\ \ &\lambda\sum_{y_0,y_1,\bm x} e^{2\pi i \left(\alpha y_0 X(\bm x) + \frac{1}{2} y_0 Q(\bm x) + \frac{1}{2} y_0y_1 + \frac{1}{2}y_1Q'(\bm x) + \beta y_1(1-X(\bm x)) + R(\bm x)\right)}\ket{\bm x_{\vec o}}\bra{\bm x_{\vec i}}\qquad\qquad\qquad\qquad{} \\
		&\reduces  2\lambda\sum_{\bm x}e^{2\pi i \left(\frac{1}{2}Q(\bm x)Q'(\bm x) + \alpha X(\bm x) \lift{Q'}(\bm x) + \beta (1-X(\bm x)) \lift{Q}(\bm x) + R(\bm x)\right)}\ket{\bm x_{\vec o}}\bra{\bm x_{\vec i}} \qquad\qquad\qquad\qquad\qquad\ \  \textsf{[Case]}
	\end{align*}
	\vspace{-0.7cm}
\caption{Pure path-sum reduction rules. Here $\lambda\in \C$, $\alpha,\beta\in\R$, $Q,Q',X:\mathbb{B}^k\rightarrow \mathbb{B}$ are Boolean
polynomials and $R:\mathbb{B}^k\rightarrow \R$ is any semi-Boolean function.
The symbol $\lift{P}$ for a Boolean polynomial $P:\mathbb{B}^n\rightarrow \mathbb{B}$ represents its \textit{lifting} $\lift{P}:\mathbb{B}^n\rightarrow \R$ that is defined inductively by $\lift{x} = x$, $\lift{PQ} = \lift{P}\cdot \lift{Q}$ and $\lift{P\oplus Q} = \lift{P} + \lift{Q} - 2\lift{PQ}$. Hence, it maps a Boolean polynomial to an integer polynomial which has the same value over the Booleans. Finally, $R[y_1\gets \lift{Q}]$ indicates that every occurrence of $y_1$ in the function $R$ is replaced by the value of $\lift{Q}(\vec x)$.
}\label{fig:rewrite}
\end{figure}
Of these rules, [Elim] is the easiest to understand: a single variable that does not occur in any other part of the expression can be safely removed. The other rules are however considerably more opaque in their interpretation. 
In the next section we will see that translated into the ZH-calculus, they gain an intuitive meaning.

\section{Translating path-sums into ZH-diagrams}
\label{s:bijection}

In this section we will see that pure path-sums can be straightforwardly represented by hypergraph-like ZH-diagrams and vice versa. We will use this correspondence to translate the path-sum reduction rules into rules for the ZH-calculus, and prove them diagrammatically. In the process we will see that these rules correspond to quite canonical hypergraph-theoretic operations.

First, let us describe the translation between ZH-diagrams and pure path-sum expressions. Recall that we write $[n] := \{1,2,\ldots,n\}$

\begin{definition}\label{def:zh-data}
A hypergraph-like ZH-diagram $D$ is given by the following data:
\begin{enumerate}
\item Finite sets $[k]$ and $[l]$ of \textit{spiders} and \textit{H-boxes}, respectively,
\item a function $H : [l] \to \mathcal P([k])$ where $H(i) = J$ when the $i$-th H-box is connected to all of the spiders in $J$,
\item a set of phase angles $\{\alpha_1, \ldots, \alpha_l\}$ where $\alpha_j \in [0, 2\pi)$ is the label on H-box $j$,
\item an input function $I : [m] \to [k]$ and an output function $O: [n] \to [k]$where $I(i) = j$ when the $i$-th input of $D$ is connected to spider $j$ and similarly  $O(i) = j$ when the $i$-th output is connected to spider $j$.
\end{enumerate}
\end{definition}

Note that the first two items in Definition~\ref{def:zh-data} specify an undirected hypergraph, whence the name.

Following \cite{PQP} let us introduce notation for the computational basis states and effects:
\tikz{\node [style=down point] (0) at (0,-0.1) {$x$};\node [style=none] (1) at (0,0.4) {}; \draw (0) to (1);}
$:= \ket{x}$ 
and \tikz{\node [style=up point] (0) at (0,0.1) {$x$};\node [style=none] (1) at (0,-0.4) {}; \draw (0) to (1);}
$:= \bra{x}$. 
Then starting with a generic hypergraph-like ZH-diagram, we can expand each of the Z-spiders as a sum over basis elements as follows:
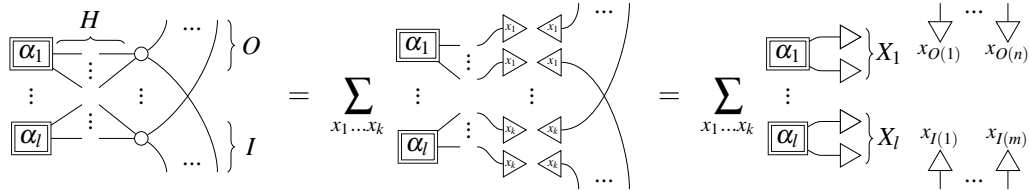
\[
\input{bij-proof.tikz}
\ \ =\ \ 
\sum_{x_1...x_k} \ \input{bij-proof-2.tikz}
\ \ =\ \ 
\sum_{x_1...x_k} \ 
\input{bij-proof-3.tikz}
\]
where the sets $X_1, \ldots, X_l$ are defined as $X_j := \{ x_i \,|\, i \in H(j) \}$. Note we can write basis elements  `sideways' without ambiguity as $(\ket x)^T = \bra x$. In the RHS above, each H-box contributes a phase $\alpha_j$ when all of the variables in $X_j$ are $1$ (i.e. when $\prod X_j = 1$) and a phase of $0$ otherwise. So the whole ZH-diagram evaluates to:
\begin{equation}\label{eq:path-sum-eval}
\sum_{x_1...x_k} e^{i \big[ \sum_j \alpha_j \cdot \prod X_j \big]}
\ket{x_{O(1)}\ldots x_{O(n)}}
\bra{x_{I(1)}\ldots x_{I(m)}}
\end{equation}
Letting $\phi(\bm x) := (\sum_j \alpha_j \cdot \prod X_j)/2\pi$, $\vec i := [I(1),...,I(m)]$, and $\vec o := [O(1),...,O(n)]$, we see that \eqref{eq:path-sum-eval} is a generic pure path-sum expression.

Conversely, from any pure path-sum expression $e$, representing the phase polynomial $\phi(\bm x)$ as a sum of monomials (cf.~Remark~\ref{remark:monomials}), we can reconstruct $H$ and $\{\alpha_j\}_j$ from the phase polynomial $\phi(\bm x)$ and the functions $I$ and $O$ from the lists $\vec i$ and $\vec o$, such that evaluating the ZH-diagram as in \eqref{eq:path-sum-eval} gives $e$. If we ensure $\phi(\bm x)$ has no repeated monomials, this reconstruction is furthermore unique, up to re-indexing of H-boxes.

Let $\zp{.}$ be the operation of translating a ZH-diagram to a pure path-sum expression by evaluation, and let $\pz{.}$ by the process of reconstructing a ZH-diagram from a pure path-sum expression. By construction these satisfy:
$\zp{\pz{e}} = e$ and
$\pz{\zp{D}} \cong D$, for any pure path-sum expressions $e$ and hypergraph-like ZH-diagrams $D$,
where `$\cong$' means equal up to permutation of the sets of spiders and H-boxes. In \cite{vilmart2020structure} it is shown that this construction is actually functorial and leads to an equivalence of categories.

\begin{example}
	The usual diagram for the CNOT gate in the ZH-calculus is given by the LHS below.
	\[\input{cnot-zh.tikz} \ = \ \input{cnot-zh-hypergraph.tikz}\]
	Here the RHS is the hypergraph-like diagram resulting from applying Lemma~\ref{lem:ZH-to-hypergraph} to the LHS.
	Translating it into a path-sum (and recalling that by convention unlabelled H-boxes have a label of $-1$) gives
	\begin{equation}\label{eq:CNOT-path-sum-ZH}
	\sum_{x_1,x_2,x_3,x_4\in \mathbb{B}} e^{2\pi i \frac12 (x_2x_3 + x_1x_3 + x_3x_4)}\ketbra{x_1x_4}{x_1x_2}.
	\end{equation}
	Translating the CNOT path-sum of Eq.~\eqref{eq:cnot-path-sum} into a ZH-diagram gives:
	\[\input{cnot-zh-hypergraph2.tikz}\]
\end{example}

By applying $\pz{.}$ to both sides of the rules of Figure~\ref{fig:rewrite}, we get an equation between (families of) ZH-diagrams. For [Elim] it is easy to see that the corresponding ZH-calculus rule is the simple removal of an arity-0 spider representing the scalar $2$. For the other rules it is harder to see directly what the translation should be. We cover each of the rules [$\omega$], [HH] and [Case] in the next subsections.

\subsection{Hyper-local complementation}\label{sec:hlc}

In this section, we will look at the [$\omega$] rule from Figure~\ref{fig:rewrite} and show it is equivalent to a new simplification we can derive using the ZH-calculus, which we call \textit{hyper-local complementation}.

\begin{definition}
	Let $G=(V,E)$ be a simple graph and $u\in V$ a vertex. The {\em local complementation of $G$ about the vertex $u$}, written as $G\star u$, is the graph $(V, E')$ where $\{v,w\}\in E'$ iff $\{v,w\}\not\in E$ when $v$ and $w$ are both neighbours of $u$ in $V$, and $\{v,w\}\in E'$ iff $\{v,w\}\in E$ otherwise. In other words: $G\star u$ is the same graph as $G$ except that neighbours of $u$ are connected iff they are not connected in $G$.
\end{definition}
Local complementation has featured in quantum information theory as it can be used to combinatorially capture local Clifford equivalence of certain stabiliser states called \textit{graph states}~\cite{NestMBQC}. It has a corresponding rewrite rule in the ZX-calculus~\cite{DuncanPerdrixGraphStates}.
Local complementation has also been used in the context of circuit simplification with the ZX-calculus~\cite{DKPW19}. In that paper it was shown that a Z-spider labelled by a phase of $\pm \frac\pi 2$ can be deleted from a ZX-diagram without changing the linear map, as long as one first performs a local complementation about the vertex. Translating the rule from \cite{DKPW19} in ZH notation, we obtain:
\begin{equation}\label{eq:local-comp}
\input{local-comp.tikz}
\end{equation}
where the right-hand side is a totally connected graph of Z spiders, connected via H-boxes. This rule can be proven in the ZX-calculus, and hence by completeness, also in the ZH-calculus.

We can extend this to \textit{hyperlocal complementation} by introducing a !-box on each of the Z-spiders at the boundary:
\begin{proposition}\label{prop:hlocal-comp}
	The following hyperlocal complementation rule holds in the ZH-calculus.
	\begin{equation}\label{eq:hlocal-comp}
	\input{hlocal-comp.tikz}
	\end{equation}
\end{proposition}
\begin{proof}
	See Appendix~\ref{sec:proofs}.
\end{proof}

\begin{remark}
	That a version of hyperlocal complementation can be implemented on a hypergraph state using local operations was first found in \cite{gachechiladze2017graphical}. A version of hyperlocal complementation in the ZH-calculus is also presented in~\cite{BK18}.
\end{remark}

\begin{remark}
As in \cite{BK18}, we adopt a `hybrid' notation mixing !-boxes with ellipses to express the hyperlocal complementation rule. This is due to a limitation of !-box notation, which is not rich enough to capture complete graphs of unbounded size~\cite{vladimirthesis}.
\end{remark}

Let us consider \textsf{[$\omega$]} in Figure~\ref{fig:rewrite} in more detail to find the connection between it and Eq.~\eqref{eq:hlocal-comp}.
The phase-polynomial on the LHS of \textsf{[$\omega$]} is $\phi(y_0,\bm x) = \frac14 y_0 + \frac12 y_0Q(\bm x) + R(\bm x)$ where $Q$ is a Boolean polynomial and $R$ is an (for us) irrelevant semi-Boolean function. The variable $y_0$ corresponds to the central spider of the LHS of Eq.~\eqref{eq:hlocal-comp} while the $\frac14 y_0$ term in the phase polynomial corresponds to the $\frac\pi2$-labelled H-box.
We can write $Q$ as $Q(\bm x) = \oplus_{j=1}^n m'_j(\bm x)$ for some monomials $m'_j$.
Note that we have the identity $e^{2\pi i \frac12 y_0 Q(\bm x)} = e^{2\pi i \frac12 \sum_j y_0 m'_j(\bm x)}$.
Hence, the term $\frac12 y_0Q(\vec x)$ contributes an H-box to the ZH-diagram for each monomial $m'_j$, which are the neighbouring H-boxes of the central spider in Eq.~\eqref{eq:hlocal-comp}.

The RHS phase polynomial in \textsf{[$\omega$]} is $\phi(\bm x) = \frac18 - \frac14\lift{Q}(\bm x) + R(\bm x)$. The function $R$ is again irrelevant, and the term $\frac18$ becomes the $e^{i\pi/4}$ scalar in the RHS ZH-diagram.
The lifting of $Q$ is given by
\begin{equation}\label{eq:lifting-Q}
\lift{Q} = \sum_{r=1}^n ~~\sum_{E \in \mathcal{P}_r([n])} (-2)^{r-1} \prod\limits_{i\in E} m'_i
\end{equation}
where $\mathcal{P}_r([n])\subseteq \mathcal{P}([n])$ is the set of subsets of $[n]$ that contain exactly $r$ elements. We can then write $\lift{Q} = \sum_j m'_j + \frac{1}{2} \sum_{j<k} m'_j m'_k + f(\bm x)$ where the first two sums represent the terms corresponding to $r=1$ and $r=2$, and $f$ contains the remaining terms.
Hence,
$$e^{2\pi i (-\frac{1}{4} \lift{Q})} = e^{2\pi i (-\frac{1}{4} \sum_j m'_j + \frac{1}{2} \sum_{j<k} m'_j m'_k + f(\bm x))} = e^{2\pi i (-\frac{1}{4} \sum_j m'_j + \frac{1}{2} \sum_{j<k} m'_j m'_k)}$$
as $f$ is valued in the integers so that it does not contribute to the phase.
Each of the $e^{-i\frac\pi2 m'_j}$ terms contributes an $-\frac\pi2$-labelled H-box
connected to the spiders of the monomial $m'_i$, while each $e^{i\pi m'_jm'_k}$ term contributes an unlabeled H-box connected to all the spiders of both $m'_j$ and $m'_k$. This indeed results in the fully connected graph in the RHS of Eq.~\eqref{eq:hlocal-comp}.

\subsection{Fourier Hyper pivot}
\label{s:hyper-pivot}

On the LHS of [HH] in Figure~\ref{fig:rewrite} the phase polynomial is $\phi(y_0,y_1,\bm x) = \frac12 y_0y_1 + \frac12 y_0 Q(\bm x) + R(y_1,\bm x)$. Hence, in the corresponding ZH-diagram we see that $y_0$ and $y_1$ are connected by an arity-2 exponentiated H-box with a phase of $2\pi \frac12 = \pi$, and hence is a regular Hadamard gate. Writing the Boolean polynomial $Q$ as $Q = \bigoplus_j^n m'_j$ where the $m'_j$ are monomials as in the previous section we see that each monomial introduces an H-box to the ZH-diagram that is connected to the spider of $y_0$.

We can separate the action of $y_1$ in $R$ as $R(\bm x,y_1) = S(\bm x)y_1 + T(\bm x)$ for some functions $S$ and $T$ where we can furthermore expand $S$ as $S(\x) = \sum_j \frac{\alpha_j}{2\pi} m_j(\x)$ for some monomials $m_j$. Hence, in the translated ZH-diagram $y_1$ shares an $\alpha_j$-valued H-box with each of the spiders of the monomials $m_j$.
Combining these observations we see that the relevant part of the pure path-sum on the LHS of [HH] is:
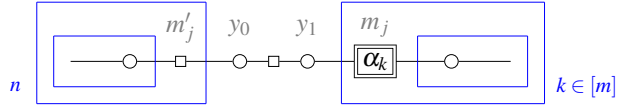
\ctikzfig{fhp-LHS}

To write the RHS of [HH] we need to represent $R[y_1 \gets \lift{Q}] = S(\bm x)\cdot \lift{Q}(\bm x) + T(\bm x)$. By representing elements of the powerset $\mathcal{P}([n])$ as bitstrings in $\mathbb{B}^n$ we can rewrite the lifting of $Q$ in Eq.~\eqref{eq:lifting-Q} as
$$\lift{Q} = \sum_{\mathbf{b}\in \B^n_*} (-2)^{\vert \mathbf{b} \vert -1} \prod\limits_{i\in [n]} (m'_i)^{b_i}.$$
Here $\prod\limits_{i\in [n]} (m'_i)^{b_i}$ is again a Boolean monomial that is $1$ precisely when $m'_i$ is $1$ for all $i$ for which $b_i=1$.

Hence, the relevant term in the RHS phase polynomial becomes
$$S \cdot \lift{Q} = \sum\limits_j \sum\limits_{\mathbf{b}\in \B^n_*} (-2)^{\vert \mathbf{b} \vert -1} \frac{\alpha_j}{2\pi} m_j \prod\limits_{i\in [n]} (m'_i)^{b_i}$$
In the translation to the ZH-calculus we hence get H-boxes with a phase of $(-2)^{\vert \mathbf{b} \vert -1} \alpha_j$ that are connected to all the spiders of $m_j$ and to all the spiders of the monomial $\prod\limits_{i\in [n]} (m'_i)^{b_i}$. We can represent this using the disconnect box of Eq.~\eqref{eq:disconnect-box}, so that the corresponding ZH-diagram is:
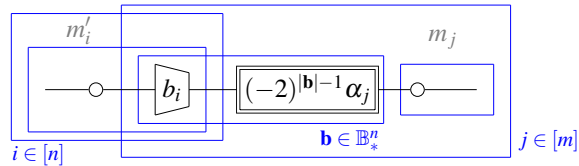
\ctikzfig{fhp-RHS}



Hence, [HH] is equivalent to the following result for ZH-diagrams, that we call the \emph{Fourier hyper-pivot}.

\begin{theorem}[Fourier hyper-pivot]\label{thm:fhp}
        The following equation holds in the ZX-calculus for any set of real numbers $\alpha_1,\ldots,\alpha_m$.
	\beq
                \input{figures/fourierHyperPivotExponential.tikz}
		\eeq
\end{theorem}

To see why we call this a Fourier hyper-pivot, let us introduce the notion of a regular pivot.









\begin{definition}
	Let $G=(V,E)$ be a simple graph and $\{u,v\}\in E$ a connected pair of vertices. The {\em pivot of $G$ along the edge $uv$} is the graph $G\star u\star v \star u$.
\end{definition}
Whereas local complementation complements the connectivity of the neighbours of a vertex, a pivot along $uv$ complements the connectivity between three groups of vertices: those connected to $u$ and not to $v$, those connected to $v$ but not to $u$, and those connected to both $u$ and $v$.

In the ZX-calculus, pivoting on a graph-like diagram can be proven using the bialgebra rule between the Z- and X-spider~\cite{DP3}. The version we are interested in is where we pivot on two connected spiders $u$ and $v$ followed by the deletion of these same vertices. This can be represented particularly elegantly using !-boxes:

\begin{proposition}\label{prop:pivot}\cite{DKPW19}
	The following pivoting rule holds in the ZH-calculus for any $n,m\in \N$.
	\[\input{pivot.tikz}\]
\end{proposition}
This is indeed a pivot (followed by vertex deletions) as every vertex (spider) connected to the left pivoted vertex becomes connected via a 2-ary H-box to every neighbour of the right pivoted vertex. The reason we only distinguish two groups of vertices, instead of three, is because spiders that are connected to both the vertices also belong to both !-boxes, and hence also get the appropriate connectivity. Since these vertices belong to both !-boxes, they furthermore get connected to \emph{themselves}, a connection that can be simplified to a simple phase:
\ctikzfig{self-loop-removal}

By allowing each H-box in Proposition~\ref{prop:pivot} to have arbitrary arity, we can generalise the above rule to a \emph{hyper-pivot} followed by two vertex deletions. 

\begin{proposition}\label{prop:hyper-pivot}
The following hyper-pivot rule holds in the ZH-calculus for any $n,m\in \N$.
        \beq\tag{RHP}\label{eq:regular-hyper-pivot}
                \input{regularHyperPivot.tikz}
        \eeq
\end{proposition}
We see that Proposition~\ref{prop:hyper-pivot} is a special case of Theorem~\ref{thm:fhp} where all the $\alpha_j$ are equal to $\pi$ (which corresponds to $R$ being a Boolean polynomial in [HH]). To prove Theorem~\ref{thm:fhp} we combine the hyper-pivot rule with the Fourier transform of Proposition~\ref{prop:fourier-transform}, hence the name. The proofs of Theorem~\ref{thm:fhp} and Proposition~\ref{prop:hyper-pivot} are given in Appendix~\ref{sec:proofs}.




\begin{remark}
	If we assume \eqref{eq:regular-hyper-pivot} as a rule of the ZH-calculus together with (ZS1), (ZS2) and (HS2) of Figure~\ref{fig:ZH-rules}, then we can prove the rules (BA1), (BA2), and (HS1). Hence, hyperpivoting supersedes these separate rules, resulting in a more `compact' calculus.
\end{remark}

\begin{remark}
	Similar to how a pivot is implemented by a combination of three local complementations, our hyperpivot rule can be implemented by three of the hyperlocal complementations of \cite{gachechiladze2017graphical}.
\end{remark}

Using the hyperpivot rule we can straightforwardly prove identities that would be hard to show using the ZX-calculus. We give an example of this in Appendix~\ref{sec:example-hpivot}.









\subsection{Case hyper pivot}
\label{s:case}

Like the \textsf{[HH]} rule, the \textsf{[Case]} rule enables the elimination of a pair of variables $y_0, y_1$ from a path-sum expression. However, unlike the case rule, \textit{both} variables are allowed to occur in monomials that have coefficients other than $\frac 1 2$, as long as they are `orthogonal' in a certain sense. That is, the non-$\frac 1 2$ monomials containing $y_0$ must be multiplied by some boolean function $X(\vec x)$, whereas those containing $y_1$ must be multiplied by its negation. This can be expressed as the following ZH-calculus rule:

\begin{theorem}[Case hyper-pivot]The following rewrite rule holds in the ZH-calculus.
	$$\scalebox{1.0}{\input{caseHyperPivotExponential.tikz}}$$
\end{theorem}

We present the details of its proof and its exact relation to the [Case] rule in Appendix~\ref{sec:case-rule}.

This rule seems less canonical than the others, and it is in fact omitted in a later presentation of the path-sum formalism by Amy~\cite{amythesis}. However, it is interesting to note that this rule essentially arises from two different, incompatible simplifications using hyper-pivoting starting from the same ZH-diagram (diagram~\eqref{eq:b} in Appendix~\ref{sec:case-rule}).
To adopt terminology from rewrite theory, the Case hyper-pivot rule arises from closing a \textit{critical pair} of the hyper-pivot rule with itself. This shows firstly that simplification using the other two laws is not confluent, which is unsurprising given its heuristic nature. More interestingly, it suggests that standard automated techniques for dealing with critical pairs, namely \textit{Knuth-Bendix completion}, could yield useful new simplification rules.

\section{Conclusion and Future Work}

We have found a bijective correspondence between path-sum expressions and ZH-diagrams and we gave ZH-calculus versions of each of the path-sum simplification rules. Furthermore, the derivation of the Case hyper-pivot rule suggests many more such rules could be recovered automatically by studying overlapping applications of the existing simplification rules.

The natural next step is to cash in these new structural insights to develop new techniques to simplify and verify circuits. ZH-diagrams and the Fourier hyper-pivot have been implemented using the PyZX library~\cite{pyzx} and seem to be effective at reducing many families of circuits to a compact form analogous to the GSLC form of ZX-diagrams~\cite{backens1}, however it is a topic of ongoing research to characterise exactly when this succeeds and when it succeeds efficiently.

Amy showed that the path-sum reduction rules suffice to verify the functional interpretation of many quantum circuits. By casting his rules in the ZH-calculus we extend this to arbitrarily constructed linear maps based on hypergraph states. In particular, we can use our results to verify and analyse MBQC schemes based on hypergraph states, such as those of Refs.~\cite{takeuchi2019quantum,gachechiladze2019changing}. Using graphical languages to study MBQC schemes has already resulted in several new results~\cite{duncan2010rewriting,backens2020extraction} and so this seems like a promising approach to new results in the burgeoning field of hypergraph MBQC.

Finally, it does not seem like the simplification rules here have yet captured the full power of the ZH-calculus. It would be interesting to see if other ZH-calculus rules, such as the ortho rule from Ref.~\cite{BK18}, can be translated into useful, and previously unconsidered simplification rules for path-sum expressions and/or ZH-diagrams.

\textbf{Acknowledgements}: We wish to thank Matthew Amy and Neil J. Ross for valuable discussions regarding path-sums, and Matt for help in understanding his software Feynman. AK and JvdW gratefully acknowledge support of AFOSR grant FA2386-18-1-4028.

\bibliographystyle{eptcs}
\bibliography{main}

\appendix



\section{Proofs of graphical rewrite rules}\label{sec:proofs}

We present here the committed proofs of graphical rewrite rules.
Note that the equation labels like (HS2) refer to the rules presented in Figure~\ref{fig:ZH-rules}.

\begin{proof}[Proof of Proposition~\ref{prop:hlocal-comp}]~
	\[\input{hlocal-comp-proof.tikz} \qedhere\]
\end{proof}

\begin{proof}[Proof of Proposition~\ref{prop:hyper-pivot}]~
	\ctikzfig{RHP-proof-compact-1}
	\ctikzfig{RHP-proof-compact-2}
	\ctikzfig{RHP-proof-compact-3}
	\[\input{RHP-proof-compact-4.tikz} \qedhere\]
\end{proof}










\begin{proof}[Proof of Theorem~\ref{thm:fhp}]
Instead of proving Theorem \ref{thm:fhp} exactly, we will prove the slightly more general following equation, that has regular instead of exponentiated H-boxes:
$$\input{figures/fourierHyperPivot.tikz}$$
Let us prove this equation:








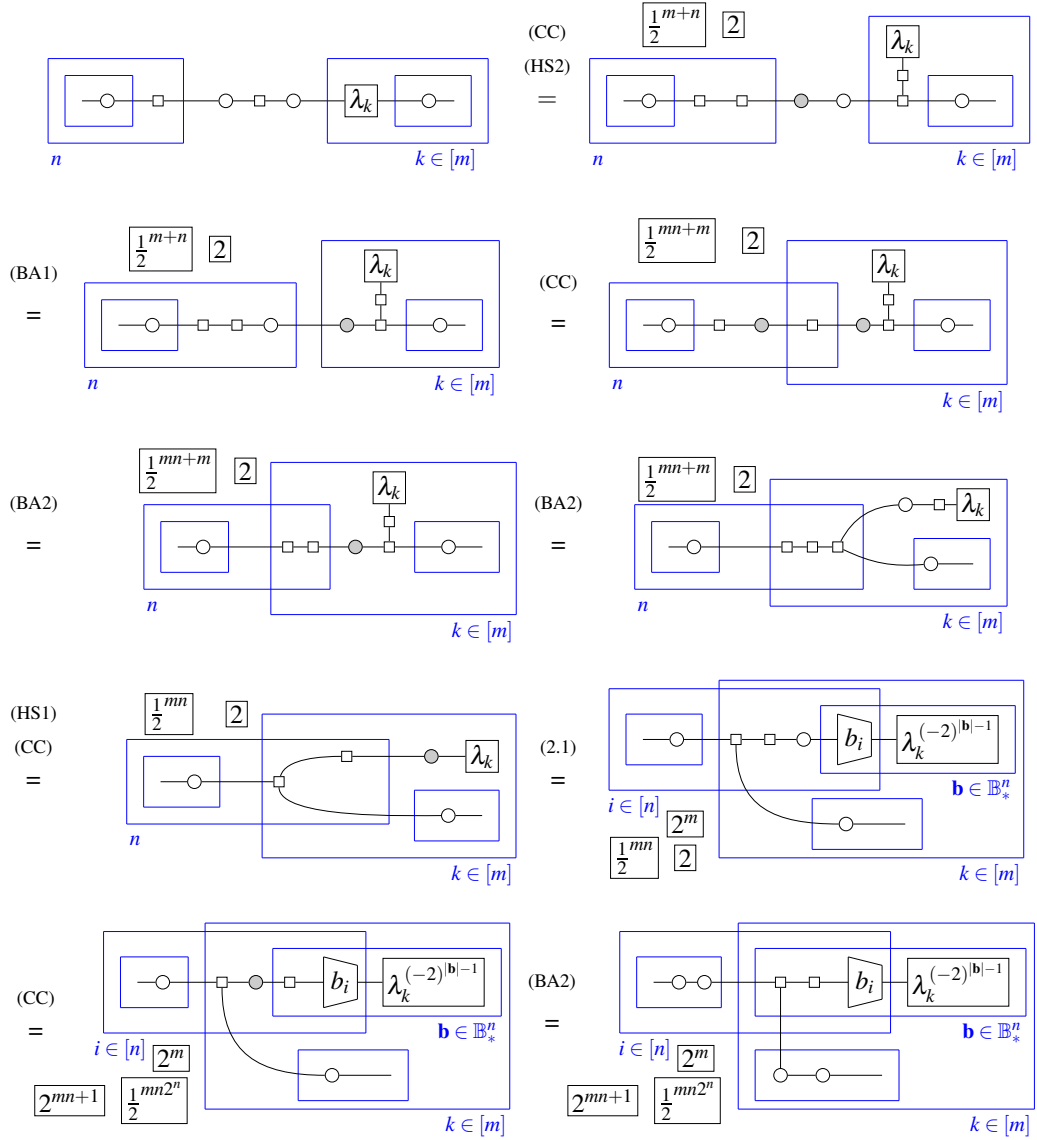
\[\input{fhp-proof-compact-1.tikz}\]
\[\input{fhp-proof-compact-2.tikz}\]
\[\input{fhp-proof-compact-3.tikz}\]
\[\input{fhp-proof-compact-4.tikz}\]
\[\input{fhp-proof-compact-5.tikz}\]

After fusing the connected white spiders, we need one final step,
the correctness of which we prove in Lemma~\ref{lem:fhp-helper}:
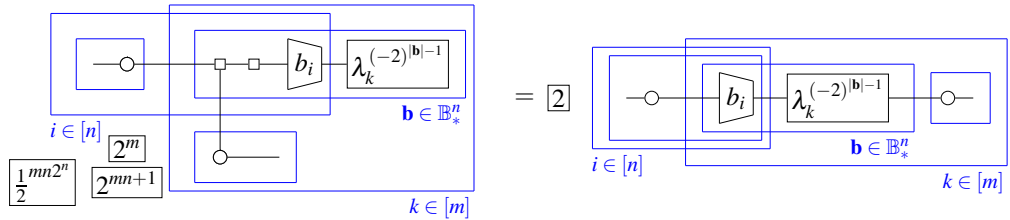
$$\input{figures/fourierHyperPivotStepFinal.tikz}$$
\end{proof}

\begin{lemma}\label{lem:fhp-helper}
	Let $\chi(\lambda,\mathbf{b}) := \lambda^{(-2)^{\vert\mathbf{b}\vert -1}}$.
	The following equation holds in the ZH-calculus for all $n$ and $m$:
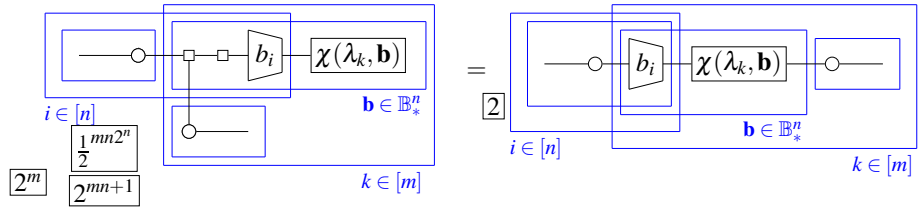
	$$\input{figures/FHPrec.tikz}$$
\end{lemma}
\begin{proof}
	We prove by induction on $m$. The base case $m=0$ is trivial, so suppose the equation holds for a fixed $m$. We will prove it for $m+1$. \\
	First, expand the !-box of $m+1$ one time:
	$$\input{figures/FHPrec1.tikz}$$
	Then by induction hypothesis:
	$$\input{figures/FHPrec2.tikz}$$
	Before stating the next step, we prove that for any Boolean $b$:
	\begin{equation}\label{eq:fhp-helper}
		\input{figures/FHPrec3.general.tikz}
	\end{equation}
	Indeed for both $b=0$ and $b=1$, the equation holds:
	$$\input{figures/FHPrec3.0.tikz}
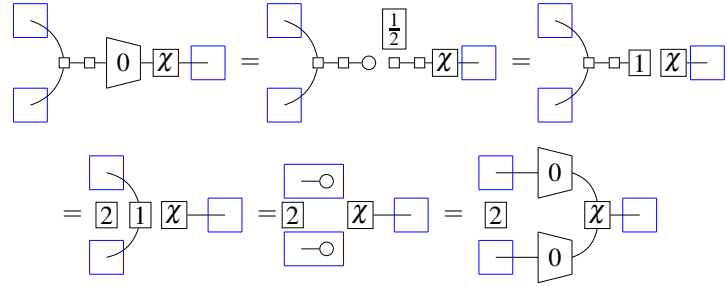
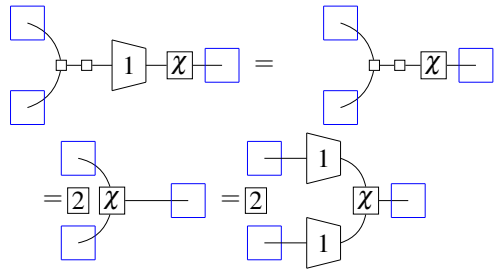
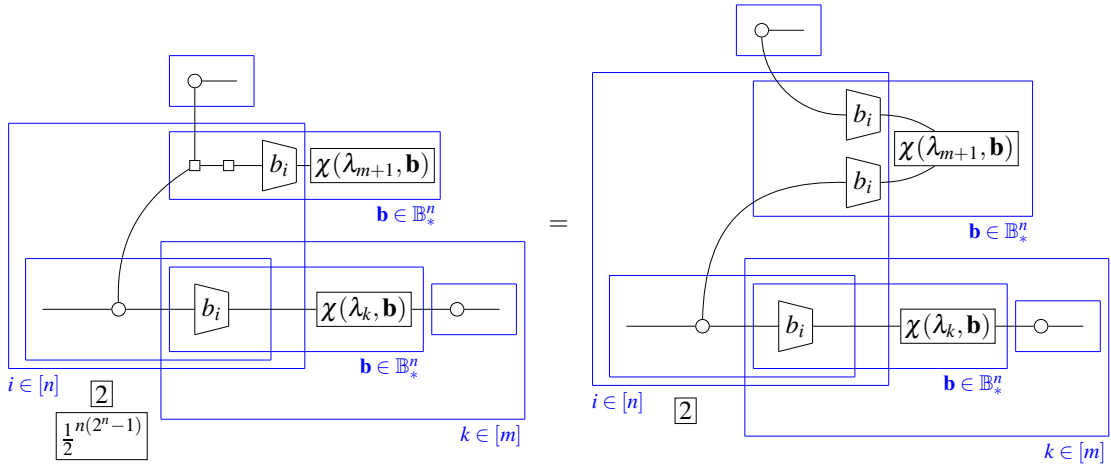
$$
	$$\input{figures/FHPrec3.1.tikz}$$
	Using Eq.~\eqref{eq:fhp-helper}:
	$$\input{figures/FHPrec3.tikz}$$
	Now note that each of the white spiders in the top !-box is connected by many wires to the H-box labelled with $\chi(\lambda_{m+1},\vec b)$. Hence, using Lemma~\ref{lem:sv} we can ignore the top $b_i$ disconnect box:
	$$\input{figures/FHPrec3prime.tikz}$$
	And finally we can put back the $m+1$ term of the !-box:
	\[\input{figures/FHPrec4.tikz} \qedhere\]
\end{proof}

\section{Usage example of the regular hyper pivot}\label{sec:example-hpivot}

Using the hyperpivot rule we can straightforwardly prove identities that would be hard to show using the ZX-calculus.
For instance, consider the following Toffoli circuit.
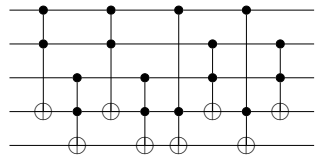
\[\input{figures/fullCircuit.tikz}\]
By reducing it to a hypergraph-like ZH-diagram, and repeatedly applying hyperpivoting to any pair of connected internal spiders, we can reduce it to the identity:
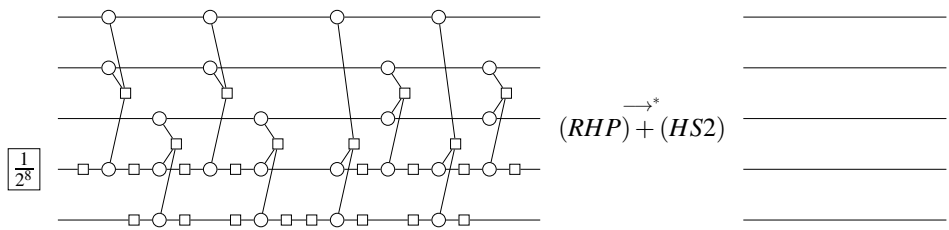
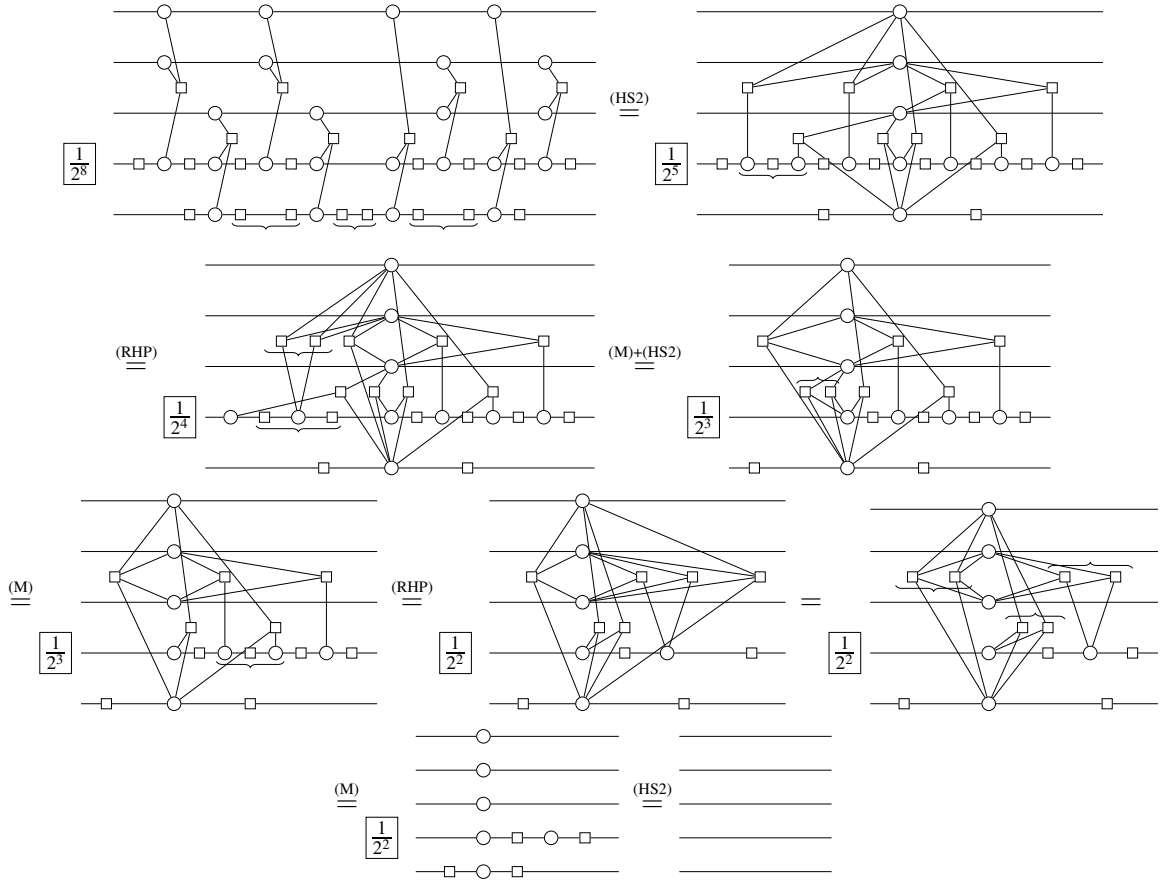
\[\input{figures/circuitExampleReduction.tikz}\]

$$\input{figures/circuitExampleDiagram.tikz} \stackrel{\text{\tiny{(HS2)}}}{=} \input{figures/circuitExampleDiagramStep2.tikz}$$ 
$$\stackrel{\text{\tiny{(RHP)}}}{=} \input{figures/circuitExampleDiagramStep3.tikz}\stackrel{\text{\tiny{(M)+(HS2)}}}{=} \input{figures/circuitExampleDiagramStep4.tikz}$$
$$\stackrel{\text{\tiny{(M)}}}{=}\input{figures/circuitExampleDiagramStep5.tikz}\stackrel{\text{\tiny{(RHP)}}}{=} \input{figures/circuitExampleDiagramStep6.tikz} = \input{figures/circuitExampleDiagramStep6.1.tikz}$$
$$\stackrel{\text{\tiny{(M)}}}{=} \input{figures/circuitExampleDiagramStep7.tikz} \stackrel{\text{\tiny{(HS2)}}}{=} \input{figures/circuitExampleDiagramStep8.tikz}$$

The several uses of (HS2) are highlighted by braces as follows: $\input{figures/hhAcc.tikz}$. 
In this case, the multiplication law (M) is only used with
$-1\times(-1)=1$, then two H-boxes connected to the same white dots eliminates themselves.
They are signaled with braces:
$$\input{figures/multiplyAcc.tikz}$$
Finally, the other operations are (RHP) applications, and highlighted this way: $\input{figures/rhpAcc.tikz}$.

\section{The [Case] rule in the ZH-calculus}\label{sec:case-rule}



The \textsf{[Case]} rule as stated in Figure~\ref{fig:rewrite} can be generalized to the following equation:
$$\sum_{y_0,y_1,\bm x} \varphi^X (-1)^{y_0Q} (-1)^{y_0y_1} (-1)^{y_1Q'} \psi^{1-X} \ket{\bm x_{\vec o}}\bra{\bm x_{\vec i}} \reduces 2 \sum_{\bm x} (-1)^{QQ'} \varphi[y_0 \leftarrow \overline{Q'}]^X \psi[y_1 \leftarrow \overline{Q}]^{1-X} \ket{\bm x_{\vec o}}\bra{\bm x_{\vec i}}$$
where $\varphi$ and $\psi$ are complex functions over$y_0$ and $\bm x$, respectively $y_1$ and $\bm x$, and
$X$, $Q$ and $Q'$ are boolean polynomials over $\bm x$.

This rule indeed generalises [Case] in Figure~\ref{fig:rewrite}, which is easily seen
by replacing $\varphi$ by $e^{2\pi i (\alpha y_0 X + R)}$ and
$\psi$ by $e^{2\pi i (\beta y_1 (1-X) + R)}$.

Let us first show the correctness of this rule algebraically.
For a fixed $\x$, we have:

$$\sum_{y_0,y_1} \varphi^X (-1)^{y_0Q} (-1)^{y_0y_1} (-1)^{y_1Q'} \psi^{1-X} =
\begin{cases}
        \sum\limits_{y_0,y_1} (-1)^{y_0Q} (-1)^{y_0y_1} (-1)^{y_1Q'} \psi \ \text{if}\ X=0 \\
        \sum\limits_{y_0,y_1} \varphi (-1)^{y_0Q} (-1)^{y_0y_1} (-1)^{y_1Q'} \ \text{if}\ X=1
\end{cases}
$$

Now applying the \textsf{[HH]} rule to both cases, this reduces to:

$$
\begin{cases}
        2 (-1)^{QQ'}\psi[y_1 \leftarrow \overline{Q}] \ \text{if}\ X=0\\
        2 \varphi[y_0 \leftarrow \overline{Q'}](-1)^{QQ'}\ \text{if}\ X=1
\end{cases}
$$

which is indeed equal to
$2 (-1)^{QQ'} \varphi[y_0 \leftarrow \overline{Q'}]^X \psi[y_1 \leftarrow \overline{Q}]^{1-X}$ as required.

Using similar reasoning as in Sections~\ref{sec:hlc} and~\ref{s:hyper-pivot} we can show that this path-sum rule is equal to the following diagrammatic rule:
$$\input{figures/caseHyperPivot.tikz}$$

Before we prove this rule in the ZH-calculus, let us state the following lemma that we will need.
\begin{lemma}\label{lem:case-helper}
	The ZH-calculus proves the following, for any set of complex numbers $\lambda_i\in S$:
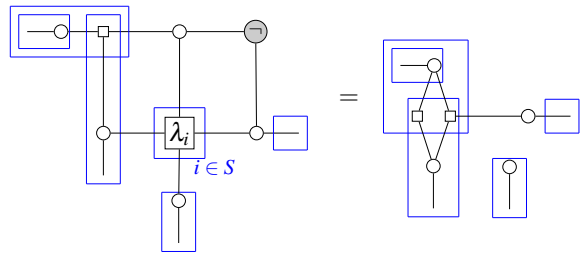
	$$\input{figures/caseLemma.tikz}$$
\end{lemma}
\begin{proof}
	Expand the NOT using its definition, and apply the hyperpivot rule to the resulting white spider and its neighbour on the left.
	It is then straightforward to check that all the resulting H-boxes with multiples of $\lambda_i$ cancel out, resulting in the correct diagram.
\end{proof}

Now let us prove the diagrammatic [Case] rule. Starting with:

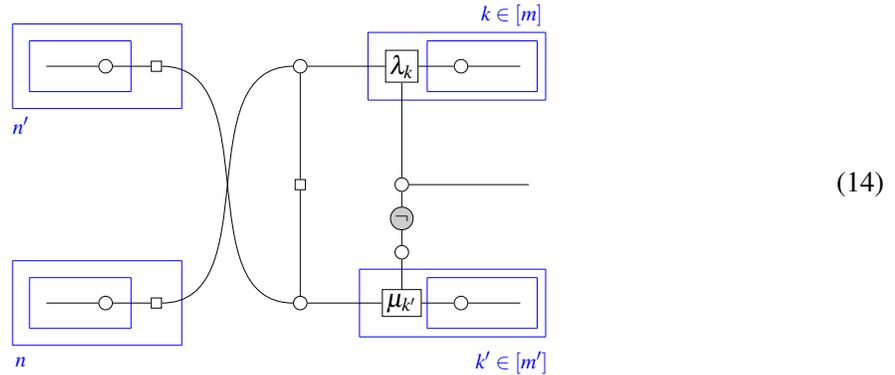
\begin{equation}\label{eq:a}
\input{figures/caseHyperPivotStep1.tikz}
\end{equation}

Unfuse all the $\mu$-labelled H-boxes using (HS1):

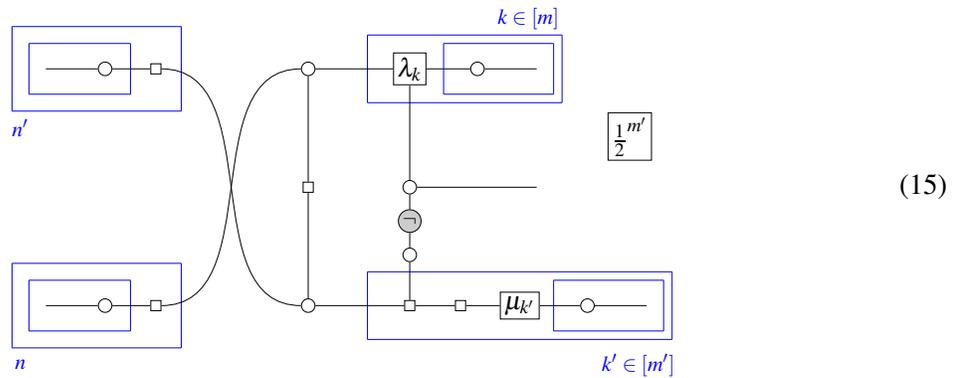
\begin{equation}\label{eq:b}
\input{figures/caseHyperPivotStep2.tikz}
\end{equation}

Apply a Fourier hyper-pivot to the middle two spiders. This results in many copies of the $\lambda_k$ H-boxes, but most of these are canceled by applying Lemma~\ref{lem:case-helper}, and we get the following diagram:

$$\input{figures/caseHyperPivotStep3.tikz}$$

Now for every instance of the pair of spiders connected by the arity-2 H-box in the $[m']$-annotated !-box we will apply a Fourier hyper-pivot. 
This results in diagram~\eqref{eq:c} below.
In order to see why we get this diagram, let us consider one of these hyper-pivots, and calculate the phases that appear on each of the H-boxes.
Expanding the $[n]$-annotated !-box we get
$n$ groups of white spiders,
$n$ `2-legged' H-boxes and $n$ `3-legged' H-boxes.
Because each pair of a `2-legged' and `3-legged' H-box coming from the same expansion of the !-box are connected to the same group of white spiders, the Fourier transform of the hyper-pivot will introduce H-boxes that have \emph{multiple} wires to the same group of white spiders. These are then collapsed to a single wire using Lemma~\ref{lem:sv}. As a result, multiple H-boxes will be connected to the same set of spiders and hence will fuse using the rule (M).

To understand the resulting phases on the H-boxes, we separate cases into H-boxes that arise just from the `2-legged' H-boxes, and ones that arise from a combination of both. The first case results in the normal pattern for a Fourier hyper-pivot and gives the lower $\mu$-labeled H-boxes in diagram~\eqref{eq:c}.

For the other case, let us denote by $f(p,q)$ the number of H-boxes after the Fourier hyper-pivot that have at least one of its connections arising from a 3-legged H-box, carry a phase of $\mu^{(-2)^{q-1}}$, and are connected to some chosen set of $p$ groups of white spiders (and the group of spiders to the right of the $\mu$ H-box, which all H-boxes will be connected to). Then when (M) is applied to fuse all the H-boxes connected to the same set of $p$ groups of spiders, the resulting H-boxes will have phases of
\begin{equation}\label{eq:mu-power}
\prod_{q=p}^{2n} \mu^{f(p,q) (-2)^{q-1}} = \mu^{\sum_{q=p}^{2n} f(p,q) (-2)^{q-1}}.
\end{equation}
Hence, to determine this phase, we need to calculate $f(p,q)$. 
So fix some set of size $p$ of groups of spiders, and suppose there is an H-box with a phase of $(-2)^{q-1}$ connected to these groups. By assumption, some of the connections `originated' from the 3-legged H-boxes. Let $1\leq r\leq p$ denote this number. Note that there are $\binom{p}{r}$ ways to choose the 3-legged H-boxes. The remaining original connections of the H-box must have come from the 2-legged H-boxes. Some of these will connect to H-boxes that aren't connected to those of the first $r$, but others will be `doubled' up. As the phase on the H-box is $(-2)^{q-1}$ there must have been a total of $q$ wires, and hence $p-q$ wires need to be doubled up. There are $r$ possible choices for doubling up, and hence there are $\binom{r}{p-q}$ possible ways we can double up the correct number of wires.
Hence:
$$f(p,q) = \sum_{r=1}^p \binom{p}{r}\binom{r}{q-p}$$
Combining this with Eq.~\eqref{eq:mu-power} we see that the resulting H-box connected to any set of $p$ groups of spiders carries a phase of 
$\mu$ raised to the power of 
$$\sum_{q=p}^{2n}\sum_{r=1}^{p} \binom{p}{r}\binom{r}{q-p} (-2)^{q-1}.$$
Let us simplify this expression:
\begingroup 
\allowdisplaybreaks
\begin{align*}
        \sum_{q=p}^{2n}\sum_{r=1}^{p} \binom{p}{r}\binom{r}{q-p} (-2)^{q-1} 
        &\stackrel{*}{=} \sum_{r=1}^{p} \binom{p}{r} \sum_{q=p}^{p+r} \binom{r}{q-p} (-2)^{q-1} \\
        &= \sum_{r=1}^{p} \binom{p}{r} \sum_{q=0}^r \binom{r}{q} (-2)^{q+p-1} \\
        &= (-2)^{p-1} \sum_{r=1}^{p} \binom{p}{r} \sum_{q=0}^r \binom{r}{q} (-2)^q \\
        &\stackrel{**}{=} (-2)^{p-1} \sum_{r=1}^{p} \binom{p}{r} (-1)^{r} \\
        &= (-2)^{p-1} \left( \sum_{r=0}^{p} \binom{p}{r} (-1)^{r} -1 \right) \\
        &\stackrel{**}{=} (-2)^{p-1} (0-1) \\
        &= -(-2)^{p-1}
\end{align*}
\endgroup
Here in the equality marked $*$ we are warranted in changing the limit of summation because $\binom{r}{q-p} = 0$ whenever $q > p+r$, and the equalities marked $**$ are applications of the binomial identity $\sum_{k=0}^n\binom{n}{k}x^ky^{n-k} = (x+y)^n$ with $y=1$.

The preceding discussions and calculations show that we indeed get the labeled H-boxes as stated in the following diagram:
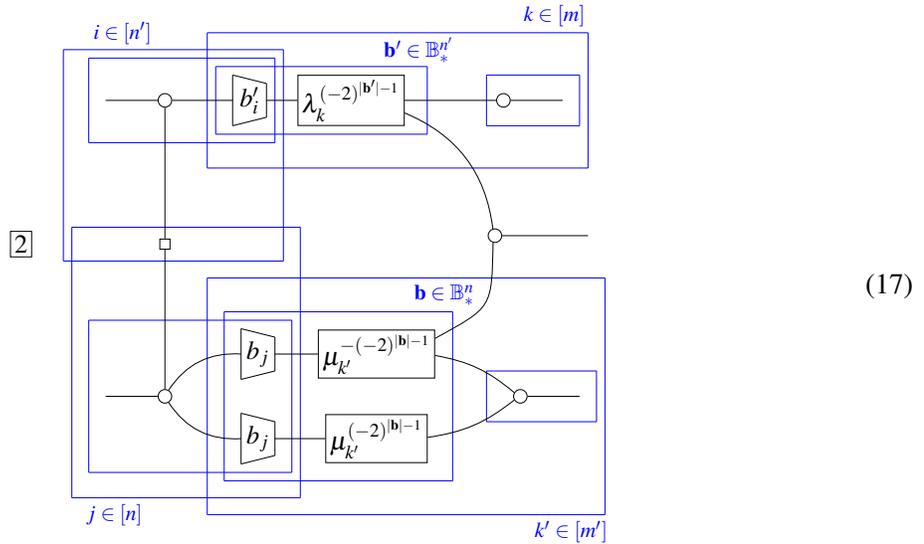
\begin{equation}\label{eq:c}
\input{figures/caseHyperPivotStep4.tikz}
\end{equation}

We claim that this diagram is equal to the following, which finishes our proof:
$$\input{figures/caseHyperPivotStep5.tikz}$$
To see that this is true, decompose the NOT in this diagram using its definition, and apply a Fourier hyper-pivot on the resulting white spider and the one beneath it.

\end{document}